\documentclass[11pt]{article}
\usepackage{parskip}
\usepackage{palatino}
\usepackage{amsmath,amssymb,amsthm,amsfonts}
\usepackage[top=8pc,bottom=8pc,left=8pc,right=8pc]{geometry}

\usepackage{setspace}
\usepackage{graphicx}
\usepackage{wrapfig}
\usepackage{epstopdf}
\usepackage{subcaption}
\usepackage{multirow}
\usepackage{nicefrac}
\usepackage{booktabs}
\usepackage{amsmath}
\usepackage{amsthm}
\usepackage[english]{babel}
\newtheorem{theorem}{Theorem}
\newtheorem{corollary}{Corollary}[theorem]
\newtheorem{lemma}[theorem]{Lemma}
\usepackage{graphicx,epstopdf}
\RequirePackage[authoryear]{natbib}
\usepackage{hyperref}
\hypersetup{colorlinks=true,citecolor=blue}
\hypersetup{linkcolor=blue}
\usepackage{algorithm}
\usepackage[noend]{algpseudocode}

\usepackage{bm}

\usepackage{array}
\newcolumntype{L}[1]{>{\raggedright\let\newline\\\arraybackslash\hspace{0pt}}m{#1}}
\newcolumntype{C}[1]{>{\centering\let\newline\\\arraybackslash\hspace{0pt}}m{#1}}
\newcolumntype{R}[1]{>{\raggedleft\let\newline\\\arraybackslash\hspace{0pt}}m{#1}}

\def\tmix{t_{\mathrm{mix}}}
\def\vrw{v_{\mathrm{RW}}}
\usepackage{array}

\include{math-commands}

\begin{document}

\begin{center}
    {\huge Scalable Estimation of Epidemic Thresholds via Node Sampling}
    
    {\Large Anirban Dasgupta and Srijan Sengupta\\ \vspace{3ex}}
    \let\thefootnote\relax\footnote{Anirban Dasgupta (anirbandg@iitgn.ac.in) is a Professor of Computer Science and Engineering of Indian Institute of Technology at Gandhinagar. His work is partially supported by grants from DBT India, Google and CISCO.
    Srijan Sengupta (sengupta@vt.edu) is an Assistant Professor of Statistics at Virginia Tech. His work is partially supported by an NIH R01 grant 1R01LM013309.}
    
\textbf{ABSTRACT}
\end{center}
Infectious or contagious diseases can be transmitted from one person to another through social contact networks.
In today's interconnected global society, such contagion processes can cause global public health hazards, as exemplified by the ongoing Covid-19 pandemic.
It is therefore of great practical relevance to investigate the network transmission of contagious diseases 
from the perspective of statistical inference.
An important and widely studied boundary condition for contagion processes over networks is the so-called \textit{epidemic threshold}.
The epidemic threshold plays a key role in determining whether a pathogen introduced into a social contact network will cause an epidemic or die out.
In this paper, we investigate epidemic thresholds from the perspective of statistical network inference.
We identify two major challenges that are caused by high computational and sampling complexity of the epidemic threshold.
We develop two statistically accurate and computationally efficient approximation techniques to address these issues under the Chung-Lu modeling framework.
The second approximation, which is based  on random walk sampling, further enjoys the advantage of requiring data on a vanishingly small fraction of nodes.
We establish theoretical guarantees for both methods and demonstrate their empirical superiority.

\clearpage
\section{Introduction}


Infectious diseases are caused by pathogens, such as bacteria, viruses, fungi, and parasites.
Many infectious diseases are also contagious, which means the infection can be transmitted from one person to another when there is some interaction (e.g., physical proximity) between them.
Today, we live in an interconnected world where such contagious diseases could spread through social contact networks to become global public health hazards.
A recent example of this phenomenon is the Covid-19 outbreak caused by the so-called novel coronavirus (SARS-CoV-2) that has spread to many countries \citep{huang2020clinical,zhu2020novel,wang2020novel,sun2020early}.
This recent global outbreak has caused serious social and economic repercussions, such as massive restrictions on movement and share market decline \citep{chinazzi2020effect}.
It is therefore of great practical relevance to investigate the transmission of contagious diseases through social contact networks from the perspective of statistical inference.

Consider an infection being transmitted through a population of $n$ individuals.
According to the susceptible-infected-recovered (SIR) model of disease spread, the pathogen can be transmitted from an infected person (I) to a susceptible person (S) with an infection rate given by $\beta$, and an infected individual becomes recovered (R) with a recovery rate given by $\mu$.
This can be modeled as a Markov chain whose state at time $t$ is given by a vector $(X^t_1, \ldots, X^t_n)$, where $X^t_i$ denotes the state of the $i^{th}$ individual at time $t$, i.e., $X^t_i \in \{S, I, R\}$.
For the population of $n$ individuals, the state space of this Markov chain becomes extremely large with $3^n$ possible configurations, which makes it impractical to study the exact system.
This problem was addressed in a series of three seminal papers by Kermack and McKendrick \citep{kermack1927contribution,kermack1932contributions,kermack1933contributions}.
Instead of modeling the disease state of each individual at at a given point of time, they proposed compartmental models, where the goal is to model the number of individuals in a particular disease state (e.g., susceptible, infected, recovered) at a given point of time.
Since their classical papers, there has been a tremendous amount of work on compartmental modeling of contagious diseases over the last ninety years \citep{hethcote2000mathematics, van2002reproduction, brauer2012mathematical}.

Compartmental models make the unrealistic assumption of homogeneity, i.e., each individual is assumed to have the same probability of interacting with any other individual. 
In reality, individuals interact with each other in a highly heterogeneous manner, depending upon various factors such as age, cultural norms, lifestyle, weather, etc.
The contagion process can be significantly impacted by heterogeneity of interactions \citep{Meyers2005, Rocha2011, Galvani2005DimensionsSuperspreading, Woolhouse1997HeterogeneitiesPrograms}, and therefore compartmental modeling of contagious diseases can lead to substantial errors.

In recent years, contact networks have emerged as a preferred alternative to compartmental models \citep{keeling2005implications}.
Here, a node represents an individual, and an edge between two nodes represent social contact between them.
An edge connecting an infected node and a susceptible node represents a potential path for pathogen transmission. 
This framework can realistically represent the heterogeneous nature of social contacts, and therefore provide much more accurate modeling of the contagion process than compartmental models.
Notable examples where the use of contact networks have led to improvements in prediction or understanding of infectious diseases include \cite{bengtsson2015using} and \cite{kramer2016spatial}.  

Consider the scenario where a pathogen is introduced into a social contact network and it spreads according to an SIR model.
It is of particular interest to know whether the pathogen will die out or lead to an epidemic. 
This is dictated by a set of boundary conditions known as the {\em epidemic threshold}, which depends on the SIR parameters $\beta$ and $\mu$ as well as the network structure itself.
Above the epidemic threshold, the pathogen invades and infects a finite fraction of the population. Below the epidemic threshold, the prevalence (total number of infected individuals) remains infinitesimally small in the limit of large networks \citep{Pastor-Satorras2015EpidemicNetworks}.
There is growing evidence that such thresholds exist in real-world host-pathogen systems, and intervention strategies are formulated and executed based on estimates of the epidemic threshold. \citep{Dallas2018ExperimentalThreshold, shulgin1998pulse,wallinga2005measles,pourbohloul2005modeling,Meyers2005}.
Fittingly, the last two decades have seen a significant emphasis on studying epidemic thresholds of contact networks from several disciplines, such as computer science, physics, and epidemiology \citep{Newman2002, Wang2003EpidemicViewpoint, colizza2007invasion, Chakrabarti2008, gomez2010discrete, wang2016predicting, wang2017unification}. 
See \cite{leitch2019toward} for a complete survey on the topic of epidemic thresholds.

Concurrently but separately, network data has rapidly emerged as a significant area in statistics.
Over the last two decades, a substantial amount of methodological advancement has been accomplished in several topics in this area, such as community detection \citep{bickel2009nonparametric,zhao2012consistency,rohe2011spectral,sengupta2015spectral}, model fitting and model selection \citep{hoff2002latent,handcock2007model,krivitsky2009representing, wang2017likelihood,yan2014model,bickel2015hypothesis, senguptapabm}, hypothesis testing \citep{ghoshdastidar2018practical,tang2017semiparametric,tang2017nonparametric,bhadra2019bootstrapbased}, and anomaly detection \citep{zhao2018performance,sengupta2018anomaly,komolafe2019statistical}, to name a few.
The state-of-the-art toolbox of statistical network inference includes a range of random graph models and a suite of estimation and inference techniques.

However, there has not been any work at the intersection of these two areas, in the sense that the problem of estimating epidemic thresholds has not been investigated from the perspective of statistical network inference.
Furthermore, the task of computing the epidemic threshold based on existing results can be computationally infeasible for massive networks.
In this paper, we address these gaps by developing a novel sampling-based method to estimate the epidemic threshold under the widely used Chung-Lu model \citep{aiello2000random}, also known as the configuration model.
We prove that our proposed method has theoretical guarantees for both statistical accuracy and computational efficiency.
We also provide empirical results demonstrating our method on both synthetic and real-world networks.

The rest of the paper is organized as follows.
In Section 2, we formally set up the problem statement and formulate our proposed methods for approximating the epidemic threshold.
In Section 3, we desribe the theoretical properties of our estimators.
In Section 4, we report numerical results from synthetic as well as real-world networks.
We conclude the paper with discussion and next steps in Section 5.

\begin{table}[h]
  \caption{Common Symbols}
  \begin{tabular}{|p{2.5cm}|p{11cm}|}
    \hline
    
    Symbol & Definition and Description \\
    
    \hline
    
    
    
    $\lambda(\mathbf{A})$ & spectral radius of the matrix $\mathbf{A}$ \\

    $d_i$ & degree of the node $i$ of the network \\
    
     $\delta_i$ & expected degree of the node $i$ of the network\\
     \hline
    
    $S(t), I(t), R(t)$ & number of susceptible ($S$), infected ($I$), and recovered/removed ($R$) individuals in the population at time $t$ \\
    
    $\beta$ & infection rate: probability of transmission of a pathogen from an infected individual to a susceptible individual per effective contact (e.g. contact per unit time in continuous-time models, or per time step in discrete-time models) \\
    
    $\mu$ & recovery rate: probability that an infected individual will recover per unit time (in continuous-time models) or per time step (in discrete-time models) \\
    
      \hline
  \end{tabular}
  \label{tbl:notation}
\end{table}

\section{Epidemic thresholds}
Consider a set of $n$ individuals labelled as $1, \ldots, n$, and an undirected network (with no self-loops) representing interactions between them.
This can represented by an $n$-by-$n$ symmetric adjacency matrix $A$, where $A(i,j) = 1$ if individuals $i$ and $j$ interact and $A(i,j) = 0$ otherwise.
Consider a pathogen spreading through this contact network according to an SIR model.
From existing work  \citep{Chakrabarti2008,gomez2010discrete,Prakash2010,wang2016predicting,wang2017unification}, we know that the boundary condition for the pathogen to become an epidemic is given by
\begin{equation}
\label{epithresh}
    \frac{\beta}{\mu} = \frac{1}{\lambda(A)},
\end{equation}
where $\lambda(A)$ is the spectral radius of the adjacency matrix $A$.

The left hand side of Equation \eqref{epithresh} is the ratio of the infection rate to the recovery rate, which is purely a function of the pathogen and independent of the network.
As this ratio grows larger, an epidemic becomes more likely, as new infections outpace recoveries.
The right hand side of Equation \eqref{epithresh} is the spectral radius of the adjacency matrix, which is purely a function of the network and independent of the pathogen.
Larger the spectral radius, the more connected the network, and therefore an epidemic becomes more likely.
Thus, the boundary condition in Equation \eqref{epithresh} connects the two aspects of the contagion process, the pathogen transmissibility which is quantified by $\beta/\mu$, and the social contact network which is quantified by the spectral radius.
If $\frac{\beta}{\mu} < \frac{1}{\lambda(A)}$, the pathogen dies out, and if $\frac{\beta}{\mu} > \frac{1}{\lambda(A)}$, the pathogen becomes an epidemic.

Given a social contact network, the inverse of the spectral radius of its adjacency matrix represents the epidemic threshold for the network.
Any pathogen whose transmissiblity ratio is greater than this threshold is going to cause an epidemic, whereas any pathogen whose transmissiblity ratio is less than this threshold is going to die out.
Therefore, a key problem in network epidemiology is to compute the spectral radius of the social contact network.

\subsection{Problem statement and heuristics}
Realistic urban social networks that are used in modeling contagion processes have millions of nodes \citep{eubank2004modelling,barrett2008episimdemics}.
To compute the epidemic threshold of such networks, we need to find the largest (in absolute value) eigenvalue of the adjacency matrix $A$.
This is challenging because of two reasons.
\begin{enumerate}
    \item First, from a computational perspective, eigenvalue algorithms have computational complexity of $\Omega(n^2)$ or higher.
    For  massive social contact networks with millions of nodes, this can become too burdensome. 
    
    \item Second, from a statistical perspective, eigenvalue algorithms require the entire adjacency matrix for the full network of $n$ individuals.
    It can be challenging or expensive to collect interaction data of $n$ individuals of a massive population (e.g., an urban metropolis). 
    Furthermore, eigenvalue algorithms typically require the full matrix to be stored in the random-access memory of the computer, which can be infeasible for massive social contact networks which are too large to be stored.
    
   \end{enumerate}
 
 The first issue could be resolved if we could compute the epidemic threshold in a computationally efficient manner.
 The second issue could be resolved if we could compute the epidemic threshold only using data on a small subset of the population.
 In this paper, we aim to resolve both issues by developing two approximation methods for computing the spectral radius.
 

To address these problems, let us look at the spectral radius, $\lambda(A)$, from the perspective of random graph models.
The statistical model is given by $A \sim P$, which is short-hand for $A(i,j) \sim \text{Bernoulli}(P(i,j))$ for $1 \le i < j \le n$.
Then $\lambda(A)$ converges to $\lambda(P)$ in probability under some mild conditions \citep{chung2011spectra,benaych2019largest,bordenave2019spectral}.
To make a formal statement regarding this convergence, we reproduce below a slightly paraphrased version (for notational consistency) of an existing result in this context.

\begin{lemma}[Theorem 1 of \cite{chung2011spectra}]
\label{lem:ChungRadcliffe2011}
Let $$
\Delta = \max_{1 \le i \le n} \sum_{j=1}^n P(i,j)
$$
be the maximum expected degree,
and suppose that for some $\epsilon>0$, 
$$
\Delta > \frac{4}{9} \log(2n/\epsilon)
$$
for sufficiently large $n$.
Then with probability at least $1-\epsilon$, for sufficiently large $n$, 
$$
|\lambda(A) - \lambda(P)| \le 2\sqrt{\Delta \log(2n/\epsilon)}.
$$
\end{lemma}

To make note of a somewhat subtle point: from an inferential perspective it is tempting to view the above result as a consistency result, where $\lambda(P)$ is the population quantity or parameter of interest and $\lambda(A)$ is its estimator.
However, in the context of epidemic thresholds, we are interested in the random variable $\lambda(A)$ itself, as we want to study the contagion spread conditional on a given social contact network.
Therefore, in the present context, the above result should not be interpreted as a consistency result.

Rather, we can use the convergence result in a different way.
For massive networks, the random variable $\lambda(A)$, which we wish to compute but find it infeasible to do so, is close to the parameter $\lambda(P)$.
Suppose we can find a random variable $T(A)$ which also converges in probability to $\lambda(P)$, and is computationally efficient.
Since $T(A)$ and $\lambda(A)$ both converge in probability to $\lambda(P)$, 
we can use $T(A)$ as an accurate proxy for $\lambda(A)$.
This would address the first of the two issues described at the beginning of this subsection.
Furthermore, if $T(A)$ can be computed from a small subset of the data, that would also solve the second issue.
This is our central heuristic, which we are going to formalize next.

\subsection{The Chung-Lu model}
So far, we have not made any structural assumptions on $P$, we have simply considered the generic inhomogeneous random graph model.
Under such a general model, it is very difficult to formulate a statistic $T(A)$ which is cheap to compute and converges to $\lambda(P)$.
Therefore, we now introduce a structural assumption on $P$, in the form of the well-known Chung-Lu model that was introduced by \cite{aiello2000random} and subsequently studied in many papers \citep{chung2002average,chung2003eigenvalues,decreusefond2012large,pinar2012similarity,zhang2017random}.
For a network with $n$ nodes, let $\mathbf{\delta} = (\delta_1, \ldots, \delta_n)'$ be the vector of expected degrees.
Then under the Chung-Lu model,
\begin{equation}
    \label{chung-lu}
    P(i,j) = \frac{\delta_i \delta_j}{\sum_{k=1}^n \delta_k}.
\end{equation}
This formulation preserves $E[d_i] = \delta_i$, where $d_i$ is the degree of the $i^{th}$ node, and is very flexible with respect to degree heterogeneity.

Under model \eqref{chung-lu}, note that $rank(P)=1$, and we have
\begin{align*}
    &P = \frac{1}{{\sum_{i=1}^n \delta_i}} \delta\delta'\\ 
 \Rightarrow \;   &P \delta 
 =   \frac{1}{{\sum_{i=1}^n \delta_i}} \delta\delta'\delta
 = \frac{\sum_{i=1}^n \delta_i^2}{{\sum_{i=1}^n \delta_i}}\delta \\
 \Rightarrow \; &\lambda(P) = \frac{\sum_{i=1}^n \delta_i^2}{{\sum_{i=1}^n \delta_i}}.
\end{align*}
Recall that we are looking for some computationally efficient $T(A)$ which converges in probability to $\lambda(P)$.
We now know that under the Chung-Lu model, $\lambda(P)$ is equal to the ratio of the second moment to the first moment of the degree distribution.
Therefore, a simple estimator of $\lambda(P)$ is given by the sample analogue of this ratio, i.e.,
\begin{equation}
 T_1(A) = \frac{\sum_{i=1}^n d_i^2}{{\sum_{i=1}^n d_i}}.   
\end{equation}
\label{eq:newmanthres}

We now want to demonstrate that approximating $\lambda(A)$ by $T_1(A)$ provides us with very substantial computational savings with little loss of accuracy.
The approximation error can be quantified as 
\begin{equation}
    \label{eq:err1}
    e_1(A) = \left|\frac{T_1(A)}{\lambda(A)}-1\right|,
\end{equation}
and our goal is to show that $e_1(A) \rightarrow 0$ in probability, while the computational cost of $T_1(A)$ is much smaller than that of $\lambda(A)$.
We will show this both from a theoretical perspective and an empirical perspective.
We next describe the empirical results from a simulation study, and we postpone the theoretical discussion to Section 3 for organizational clarity.

We used $n=5000, 10000$, and constructed a Chung-Lu random graph model where $P(i,j) = \theta_i \theta_j$.
The model parameters $\theta_1, \ldots, \theta_n$ were uniformly sampled from $(0, 0.25)$.
Then, we randomly generated 100 networks from the model, 
and computed $\lambda(A)$ and $T_1(A)$.
The results are reported in Table \ref{tbl:t1A}.
Average runtime for the moment based estimator, $T_1(A)$, is only 0.07 seconds for $n=5000$ and 0.35 seconds for $n=10000$, whereas for the spectral radius, $\lambda(A)$, it is 78.2 seconds and 606.44 seconds respectively, which makes the latter 1100-1700 times more computationally burdensome.
The average error for $T_1(A)$ is very small, and so is the SD of errors.
Thus, even for moderately sized networks where $n=5000$ or $n=10000$, using $T_1(A)$ as a proxy for $\lambda(A)$ can reduce the computational cost to a great extent, and the corresponding loss in accuracy is very small.
For massive networks where $n$ is in millions, this advantage of $T_1(A)$ over $\lambda(A)$ is even greater; however, the computational burden for $\lambda(A)$ becomes so large that this case is difficult to illustrate using standard computing equipment.

Thus, $T_1(A)$ provides us with a computationally efficient and statistically accurate method for finding the epidemic threshold.

\begin{table}[h]
\centering{
  \caption{Computational efficiency and statistical accuracy of $T_1(A)$
  \label{tbl:t1A}}
  \begin{tabular}{|r|r|r|r|r|}
    \hline
    
   $n$ & Mean Time for $\lambda(A)$ & Mean Time for $T_1(A)$ & Mean Error & SD of Error  
    \\
    \hline
    5000 & 78.20 seconds & 0.07 seconds & 0.0012 & 0.0003  
    \\
    \hline
   10000 & 606.44 seconds & 0.35 seconds & 0.0005 & 0.0002  
    \\
     \hline
      \end{tabular}}
  \end{table}

\subsection{Sampling based approximation}
The first approximation, $T_1(A)$, provides us with a computationally efficient method for finding the epidemic threshold.
This addresses the first issue pointed out at the beginning of Section 2.1.
However, computing $T_1(A)$ requires data on the degree of all $n$ nodes of the network.
Therefore, this does not solve the second issue pointed out at the beginning of Section 2.1.
We now propose a second alternative, $T_2$, to address the second issue. 
The idea behind this approximation is based on the same heuristic that was laid out in Section 2.2.
Since $\lambda(P)$ is a function of degree moments, we can estimate these moments using observed node degrees.
In defining $T_1(A)$, we used observed degrees of all $n$ nodes in the network.
However, we can also estimate the degree moments by considering a small sample of nodes, based on random walk sampling.
The algorithm for computing $T_2$ is given in Algorithm \ref{alg:rw}.

\begin{algorithm}[h]
\caption{RandomWalkEstimate}\label{alg:rw}
\begin{algorithmic}[1]
\Procedure{Estimate}{$G, r, t^*$}
\label{alg:rwestimate}
\State $x \leftarrow 1.$
\While {$t \le t^* $}
    \State $x \leftarrow $ random neighbor of $x$, chosen uniformly.
\EndWhile
\State $v \leftarrow 0$.
\While{$i \leq r$}
\State $v = v + d_{x}$
\State $x \leftarrow $ random neighbor of $x$, chosen uniformly.
\EndWhile
\State return $T_2 = v / r$.
\EndProcedure

\end{algorithmic}
\end{algorithm}

Note that we only use $(t^*+r)$ randomly sampled nodes for computing $T_2$, which implies that we do not need to collect or store data on the $n$ individuals.
Therefore this method overcomes the second issue pointed out at the beginning of Section 2.1.
The approximation error arising from this method can be defined as
\begin{equation}
    \label{eq:err2}
    e_2(A) = \left|\frac{T_2(A)}{\lambda(A)}-1\right|,
\end{equation}
and we want to show that $e_2(A) \rightarrow 0$ in probability, while the data-collection cost of $T_2(A)$ is much less than that of $T_1(A)$.
In the next section, we are going to formalize this.

\section{Theoretical results on approximation errors}
In this section, we are going to establish that the approximation errors $e_1(A)$ and $e_2(A)$, defined in Equations \eqref{eq:err1} and \eqref{eq:err2}, converge to zero in probability.
From Theorem 2.1 of \cite{chung2003eigenvalues}, we know that when \begin{equation}
\label{eq:chunglucond}
    \frac{\sum_i \delta_i^2}{\sum_i \delta_i} > 
    \log(n) \sqrt{\max_{1 \le i \le n} \delta_i}
\end{equation}
holds, then for any $\epsilon >0$,
$$
P\left[\left|\frac{\lambda(A)}{\lambda(P)}-1\right| > \epsilon\right]
\rightarrow 0.
$$
Therefore, under \eqref{eq:chunglucond}, it suffices to show that, for any  $\epsilon >0$,
$$
P\left[\left|\frac{T_1(A)}{\lambda(P)}-1\right| > \epsilon\right]
\rightarrow 0,
\text{ and }
P\left[\left|\frac{T_2(A)}{\lambda(P)}-1\right| > \epsilon\right]
\rightarrow 0.
$$

\subsection{Convergence of $T_1(A)$}
First, consider $T_1(A) = \frac{\sum_{i=1}^n d_i^2}{{\sum_{i=1}^n d_i}}$, and recall that $\lambda(P) = \frac{\sum_{i=1}^n \delta_i^2}{{\sum_{i=1}^n \delta_i}}.$
For notational convenience, define $m_1 = \sum_{i=1}^n d_i, m_2 = \sum_{i=1}^n d_i^2, \mu_1 = \sum_{i=1}^n \delta_i, \mu_2 = \sum_{i=1}^n \delta_i^2$.
We would like to show that, under reasonable conditions, for any $\epsilon >0$,
\begin{equation}
P\left[\left|
\frac{m_2 \mu_1}
{m_1 \mu_2}
-1\right| > \epsilon\right]
\rightarrow 0.
\label{m1m2_1}
\end{equation}
We will show that for any $\epsilon' > 0$,
\begin{equation}
    P\left[\left|\frac{m_1}{\mu_1} - 1\right| > \epsilon'\right]
\rightarrow 0,
P\left[\left|\frac{m_2}{\mu_2} - 1\right| > \epsilon'\right]
\rightarrow 0.
\label{m1m2_2}
\end{equation}
We first prove that $\eqref{m1m2_2}$ implies $\eqref{m1m2_1}$.
Equation $\eqref{m1m2_2}$ implies that
$$
P\left[\left\{\left|\frac{m_1}{\mu_1} - 1 \right| > \epsilon'\right\} \cup \left\{\left|\frac{m_2}{\mu_2} - 1\right| > \epsilon'\right\}\right]
\rightarrow 0.
$$
Now, consider the event
$\left\{\left|\frac{m_1}{\mu_1} - 1 \right| \le \epsilon'\right\} \cap \left\{\left|\frac{m_2}{\mu_2} - 1\right| \le \epsilon'\right\}$.
Note that $m_2/m_1$ is a strictly increasing function of $m_2$ and a strictly decreasing function of $m_1$.
Therefore, for outcomes belonging to the above event, 
$$
\frac{\mu_2}{\mu_1} \times \frac{1-\epsilon'}{1+\epsilon'}
\le
\frac{m_2}{m_1}
\le
\frac{\mu_2}{\mu_1} \times \frac{1+\epsilon'}{1-\epsilon'}.
$$
Note that
$$
1 - \frac{1-\epsilon'}{1+\epsilon'} = 
\frac{2\epsilon'}{1+\epsilon'}
< 2\epsilon',
\text{ and }
\frac{1+\epsilon'}{1-\epsilon'}-1
=
\frac{2\epsilon'}{1-\epsilon'}
< 4\epsilon',
$$
given that $\epsilon' < 1/2$.
Now, fix $\epsilon >0$ and let $\epsilon' = \epsilon/4$.
Then,
$$
\eqref{m1m2_2}
\Rightarrow
P\left[\left|
\frac{m_2 \mu_1}
{m_1 \mu_2}
-1\right| > 4\epsilon'\right]
\rightarrow 0
\Rightarrow \eqref{m1m2_1}.
$$
Thus, proving \eqref{m1m2_2} is sufficient for proving \eqref{m1m2_1}.

Next, we state and prove the theorem which will establish \eqref{m1m2_2}.

\begin{theorem}
If the average of the expected degrees goes to infinity, i.e.,
$
\frac{1}{n}{\sum_i \delta_i} \rightarrow \infty
$, and the spectral radius dominates $\log^2(n)$, i.e., $\frac{\sum_i \delta_i^2}{\sum_i \delta_i} = \omega(\log^2 n)$,
then for any $\epsilon > 0$,
$$
    P\left[\left|\frac{m_1}{\mu_1} - 1\right| > \epsilon\right]
\rightarrow 0, \text{ and }
P\left[\left|\frac{m_2}{\mu_2} - 1\right| > \epsilon\right]
\rightarrow 0.
$$

\end{theorem}
\begin{proof}
We will use Hoeffding's inequality \citep{hoeffding1994probability} for the first part, and we begin by stating the inequality for the sum of Bernoulli random variables.
Let $B_1, \ldots, B_m$ be $m$ independent (but not necessarily identically distributed) Bernoulli random variables, and $S_m = \sum_{i=1}^m B_i$.
Then for any $t > 0$,
$$
P[|S_m - {E}[S_m]| \geq t]
\le 
2 \exp\left({\frac{-2t^2}{m}}\right).
$$
In our case, 
$$
m_1 = \sum_{i=1}^n d_i
= \sum_{i=1}^n \sum_{j=1}^n A(i,j)
= 2 \sum_{i<j} A(i,j),
$$
and we know that $\{A(i,j): 1 \le i < j \le n\}$ are independent Bernoulli random variables. 
Fix $\epsilon > 0$ and note that $E[\sum_{i<j} A(i,j)] = \frac{1}{2}\mu_1$.
Using Hoeffding's inequality with $S_m = m_1/2$, $m = {n \choose 2}$, and $t = \frac{\epsilon}{2} \mu_1$,
we get
$$
P\left[\left|\frac{m_1}{2} - \frac{\mu_1}{2}\right| >
\frac{\epsilon}{2} \mu_1\right]
\le
2 \exp\left(-\epsilon^2
{\frac{\mu_1^2}{n(n-1)}}\right).
$$
Since $
\frac{1}{n}{\sum_i \delta_i} \rightarrow \infty
$, the right hand side goes to zero.
Therefore,
$$
P\left[\left|\frac{m_1}{\mu_1} - 1\right| > \epsilon\right]
\rightarrow 0.
$$
For the second part, we can characterize $m_2$ as following.  
$$E[m_2] = E[\sum_i d_i^2] = \sum_i (E[d_i])^2 + var(d_i) = \mu_2 + var(d_i),$$
and hence,
$$|m_2 - \mu_2 | \le |m_2 - E[m_2]| + |E[m_2] - \mu_2|. $$
We show that, under the given assumptions, with probability $1 - o(1)$, $|m_2 - E[m_2]|= o(\mu_2).$ Furthermore, $|E[m_2] - \mu_2| = o (\mu_2).$

As noted before, each $d_i$ is a sum of binomial random variables. By applying Chernoff-Hoeffding bound, and union bounding over all $i\in \{1, \ldots, n\}$, we can get, with probability $1 - o(1)$, and for any fixed $\epsilon \in (0,1)$,
$$ \forall i \in \{1, \ldots, n\},\ d_i \le \delta_i + \max\{\epsilon \delta_i, O(\log(n))\}. $$
Let the above event be called the event $\cal A$. If the event $\cal A$ happens, then,
\begin{align*}
   m_2 = \sum_i d_i^2& \le \sum_i \delta_i^2 + 2\delta_i\max(\epsilon \delta_i, \log n)+ \max(\epsilon^2 \delta_i^2, \log^2 n)\\
   & \le \mu_2 + 2\sum_i \delta_i(\epsilon \delta_i +  \log n) + (\epsilon^2 \delta_i^2 + \log^2 n)\\
   & \le \mu_2 + 3\epsilon\mu_2 + (n + \sum_i \delta_i)\log^2 n \\
   \left| \frac{m_2}{\mu_2} - 1 \right| & \le 3\epsilon +  (n + \sum_i \delta_i)\log^2 n / \mu_2
\end{align*}
Note that $\frac{n}{\mu_2} = \frac{1}{\sum_i \delta_i^2 / n} \rightarrow 0$ under the given assumption. Furthermore, 
\begin{align*}
    \frac{(\sum_i \delta_i)\log^2 n}{\sum_i \delta_i^2} = o(1) \rightarrow 0.
\end{align*}
Putting these together, and using $\epsilon' = 3\epsilon$ we have the given claim.

\end{proof}

Thus, we have proved that the approximation error for $T_1(A)$ goes to zero in probability.
we have already observed in Section 2.2 that the runtime for $T_1(A)$ is orders of magnitude faster that the runtime for $\lambda(A)$.
Therefore, $T_1(A)$ is both computationally efficient and statistically accurate as an approximation of the epidemic threshold.

\subsection{Convergence of $T_2(A)$}
Next, consider Algorithm \ref{alg:rw}.
Let $\pi$ denote the stationary distribution of the simple random walk on the given graph. 
Suppose the number of edges in the given graph is $m$.
Recall that, $\pi$ is given by $\pi_v = \frac{d_v}{\sum_v d_v}$ for all $v$.
For brevity, we define the mixing time of the graph $A$, denoted as $\tmix(A)$, to mean the number of steps required by the simple random walk to reach a distribution $\hat{\pi}$ such that $\|\hat{\pi} - \pi\|_1 = o(\frac{1}{n^2})$.
Let $T_2(A)$ be the estimate returned by the Algorithm \ref{alg:rw}. We first show an easy lemma that characterizes the bias of the estimator $T_2(A)$.
\begin{lemma}
\label{lem:unbiased}
If $x$ is a node that is randomly sampled from $\pi$, and $d_x$ is its degree, then $E[d_x]= \frac{\sum_i d_i^2}{\sum_i d_i}. $ Consequently if $\hat{\pi}$ is such that $\|\pi - \hat{\pi}\|_1 = o(n^{-2})$ and $x$ is sampled from $\hat{\pi}$, then $E[d_x]= (1 \pm o(1))\frac{\sum_i d_i^2}{\sum_i d_i}.$
\end{lemma}

\begin{proof}
It is easy to see that
$$ E_{x\sim \pi} [d_x] = \sum_{v = 1}^{n} d_v \times \pi_v =\frac{\sum_{v} d_v^2}{\sum_v d_v}.$$
We show the second claim as follows: 
$$ |E_{x\sim \pi} [d_x] - E_{x\sim \hat{\pi}} [d_x]| \le  \sum_{v = 1}^{n} d_v  |\pi_v - \hat{\pi_v}| \le n \|\pi - \hat{\pi}\|_1 = o(1).$$

\end{proof}
Next, we show that the estimator $\vrw$ is actually concentrated around its expectation. 


\begin{theorem} [\cite{lezaud1998chernoff}]
\label{thm:lezaud}
Let $(X_n)$ be a irreducible and reversible Markov Chain on a finite set $V$
with $Q$ being the transition matrix. Let $\pi$ be the stationary distribution.
%
Let $f: V\rightarrow \Re$ be such that $E_\pi[f] = 0$, $\|f\|_\infty \leq 1$ and $0 < E_\pi[f^2]\le b^2$. Then, for any initial distribution $q$, any positive integer $r$ and all $0< \gamma \le 1$,

\begin{align*}
  \Pr_q \left[ r^{-1}\sum_{i=1}^r f(X_i)\ge \gamma \right] \le e^{-\varepsilon(P)/5} S_q \exp\left( -\frac{r\gamma^2\varepsilon(P)}{4b^2(1 + h(5\gamma/b^2))} \right),
\end{align*}

where $\varepsilon(Q) =1 - \lambda_2(Q)$, $\lambda_2(Q)$ being the second largest eigenvalue of $P$,  $S_q = \|q/\pi\|_2$ (in the $\ell_2(\pi)$ norm). 

$$ h(x) = \frac{1}{2}(\sqrt{1+x} - (1 - x/2)).$$

If $\gamma \ll b^2$ and $\varepsilon(P) \ll 1$, the bound is

$$ (1 +o(1))S_q \exp\left( - \frac{r\gamma^2\varepsilon(p)}{4b^2(1+o(1))} \right).$$

\label{thm:markovconc}
\end{theorem}

Using the above result, we bound the sample complexity of our estimator. We first quote the following result that we use to bound $\lambda_1$ of the transition matrix. 

\begin{theorem} Let $\epsilon, \delta\in (0, 1)$. Algorithm~\ref{alg:rwestimate}, using $r = \frac{1}{\varepsilon(Q) \epsilon^{3/2}} \times \frac{12m d_{\max}}{(\sum_v d_v^2)} \log(1/\delta)$ and $t^* \ge \tmix(G)$ returns an estimate $\vrw$ that satisfies, with probability $1 - \delta,$
$$ (1 - \epsilon)\frac{\sum_v d_v^2}{\sum_v d_v}\le T_2(A) \le  (1 + \epsilon)\frac{\sum_v d_v^2}{\sum_v d_v}. $$
The number of nodes that are touched by algorithm is $O(t^{*}  + r)$. 
\end{theorem}
\begin{proof}
In our setting the set $V$ is the set of vertices. Define the function $f(X_i)$ as : $$ d_{\max} \times f(X_i) = d_{X_i} - E_{\pi}[d_{X_i}].$$
$f(\cdot)$ clearly satisfies $E_{\pi}[f] = 0$ and that $\|f\|_\infty \le 1$. We can bound $E_{\pi}[f^2]$ as
\begin{align*}
    E_{\pi}[f^2] \le d_{\max}^{-2}E_{\pi}[d_v^2] = d_{\max}^{-2} \sum_v \frac{d_v^2 \times d_v}{\sum_v d_v} =d_{\max}^{-2} \sum_v \frac{d_v^3}{\sum_v d_v} 
\end{align*}

Using the first $t^*$ steps, we reach the distribution $\hat{\pi}$ that satisfies $\|\pi - \hat{\pi}\|_1 = o(n^{-2})$. Hence,
\begin{align*}
    \|\hat{\pi} / \pi\|_2^2 & = \sum_v \pi_v(\hat{\pi}_v/\pi_v)^2 = \sum_v \hat{\pi}_v^2 / \pi_v = \sum_v(\pi_v + (\hat{\pi}_v - \pi_v))^2 / \pi_v\\
    & = \sum_v (\pi_v + 2 (\hat{\pi}_v - \pi_v) + (\hat{\pi}_v - \pi_v)^2/\pi_v)\\
    & = 1 + 2\times (1 - 1) + \sum_v (\hat{\pi}_v - \pi_v)^2/\pi_v \le 1 + \|\pi - \hat{\pi}\|_2^2 / \min(\pi_v)\\
    & \le 1 + \|\pi - \hat{\pi}\|_1^2 (\sum_v d_v)/d_{\min}
    = 1 + o(1).
\end{align*}
where the last step follows as $\|\pi - \hat{\pi}\|_1 = o(n^{-2})$.

We use $b^2 = d_{\max}^{-2} \sum_v \frac{d_v^3}{\sum_v d_v}$ and $\gamma = \epsilon d_{\max}^{-1}\times \frac{\sum_v d_v^2}{\sum_v d_v}$. Hence
\begin{align*}
    \gamma / b^2 = \epsilon d_{\max} \frac{\sum_v d_v^2 }{\sum_v d_v^3} \ \mathrm{and}\  
    \gamma^2 / b^2 = \epsilon^2 \frac{(\sum_v d_v^2)^2}{(\sum_v d_v) (\sum_v d_v^3)}
\end{align*}
Hence,
\begin{align*}
    h(5\gamma / b^2)& = \left( 1 + 5 \epsilon d_{\max} \frac{\sum_v d_v^2 }{\sum_v d_v^3}\right)^{1/2} - 1 +  5\epsilon d_{\max} \frac{\sum_v d_v^2 }{2\sum_v d_v^3} \le \left(5 \epsilon d_{\max} \frac{\sum_v d_v^2 }{\sum_v d_v^3}\right)^{1/2} + 2.5 \epsilon d_{\max} \frac{\sum_v d_v^2 }{\sum_v d_v^3} \\
    & \le 6\epsilon^{1/2} d_{\max} \frac{\sum_v d_v^2 }{\sum_v d_v^3}
\end{align*}

Plugging this, we get that 
\begin{align*}
        \frac{r\gamma^2 \varepsilon(Q)}{4b^2( 1 + h(5\gamma/b^2))} \ge r \varepsilon(Q)\times \epsilon^2 \frac{(\sum_v d_v^2)^2}{(\sum_v d_v) (\sum_v d_v^3)} \times (1 + 6\epsilon^{1/2} d_{\max} \frac{\sum_v d_v^2 }{\sum_v d_v^3})^{-1} \ge \frac{r \varepsilon(Q)\epsilon^{3/2}(\sum_v d_v^2)}{6(\sum_v d_v) d_{\max}}
\end{align*}
Setting $r = \frac{1}{\varepsilon(Q) \epsilon^{3/2}} \times \frac{6(\sum_v d_v) d_{\max}}{(\sum_v d_v^2)} \log(1/\delta)$, and using Theorem~\ref{thm:lezaud}, we can claim that, with probability $1 - \delta$, 
$$ T_2(A) \in (1\pm \epsilon) \frac{\sum_v d_v^2}{\sum_v d_v}.$$
The bound on the number of nodes touched/queried by the algorithm follows naturally. 
\end{proof}
Note that $Q = D^{-1}A$ has the same set of eigenvalues as the matrix $D^{-1/2}A D^{-1/2}$. For the Chung-Lu model, the eigenvalues of
the matrix $L= I - D^{-1/2}A D^{-1/2}$ can be bounded by the following result from \cite{chung2003eigenvalues}.

\begin{theorem}
Let $L = I - D^{-1/2}AD^{-1/2}$ denote the normalized Laplacian. Let $A$ be a random graph generated from the given expected degrees model, with expected degrees $\{\delta_i\}$, if the minimum expected degree $\delta_{\min}$ satisfies $\delta_{min} \gg \ln(n)$, then with probability at least $1-1/n= 1-o(1)$, we have that for all eigenvalues $\lambda_k(L) > \lambda_{\min}(L)$ of the 
Laplacian of $G$, 
$$\left| 1 - \lambda_k(L)\right| < 2 \sqrt{\frac{6\ln(2n)}{\delta_{\min}}} = o(1). $$
\end{theorem}
It follows above that $\varepsilon(Q) = 1 - \lambda_2(Q) =1 -  \lambda_2(D^{-1/2}AD^{-1/2}) = \lambda_{n-1}(I - D^{-1/2}AD^{-1/2}) = 1 - o(1)$. Putting these together, we get the following corollary on the total number of node queries.
\begin{corollary}
For a graph generated from the expected degrees model, with probability $1 - 1/n$,  Algorithm~\ref{alg:rwestimate}, 
needs to query $$\ln(n) + \frac{1}{\epsilon^{3/2}} \times \frac{6(\sum_v d_v) d_{\max}}{(\sum_v d_v^2)} \log(1/\delta)$$ nodes in order to get a $(1\pm \epsilon)$ estimate of $\sum_v d_v^2/2m.$
\end{corollary}

Note $\frac{6(\sum_v d_v) d_{\max}}{(\sum_v d_v^2)} \le \frac{6d_{\max}}{d_{\min}}$, but this is a loose bound, better bounds can be derived for power law degree distributions, for instance. 

Thus, we have proved that the approximation error for $T_2(A)$ goes to zero in probability.
In addition, Corollary 6.1 shows that the number of nodes that we need to query in order to have an accurate approximation is much smaller than $n$.
Furthermore, computing $T_2$ only requires node sampling and counting degrees, and therefore the runtime is much smaller than eigenvalue algorithms.
Therefore, $T_2(A)$ is a computationally efficient and statistically accurate approximation of the epidemic threshold, while also requiring a much smaller data budget compared to $T_1(A)$.

\section{Numerical results}
In this section, we characterize the empirical performance of our sampling algorithm on 
two synthetic networks, one generated from the Chung-Lu model and the second generated from the preferential attachment model of \cite{Barabasi1999EmergenceNetworks}.

\begin{table}[h]
\centering
\begin{tabular}{|l|c|c|c|c|}
\hline
Data & Nodes  & Edges  & $\lambda (A)$ & $T_1(A)$  \\
\hline
\hline
Chung-Lu & $50k$ & $72k$ & $43.83$ & $48.33$  \\
Pref-Attach & $50k$ & $250k$ & $37$ & $32.8$  \\
\hline
\end{tabular}
\caption{Statistics of the two synthetic datasets used.}
\label{tab:stats}
\end{table}

\subsection{Data}
Our first dataset is a graph generated from the Chung-Lu model of expected degrees. We generated a powerlaw sequence (i.e. fraction of nodes with degree $d$ is proportion to $d^{-\beta}$) with exponent $\beta = 2.5$ and then generated a graph with this sequence as the expected degrees. Table~\ref{tab:stats} notes that, as expected, the first eigenvalue $\lambda_1(A)$ is close to $\frac{\sum_v d_v^2}{\sum_v d_v}$.

The second dataset is generated from the preferential attachment model \citep{barabasi1999emergence}, where each incoming node adds 5 edges to the existing nodes, the probability of choosing a specific node as neighbor being proportional to the current degree of that node. While the preferential attachment model naturally gives rise to a directed graph, we convert the graph to an undirected one before running our algorithm. It is interesting to note that even in this case the Chung-Lu model does not hold, our first approximation, $T_1(A)$, is close to $\lambda(A)$.


\subsection{Implementation Details}
In each of the networks, the random walk algorithm presented in Algorithm \ref{alg:rw} was used for sampling. The random walk was started from an arbitrary node and every $10^{th}$ node was sampled (to account for the mixing time) from the walk. These samples were then used to calculate $T_2(A)$. This experiment was repeated 10 times. These gave estimates $T_2^{1},\ldots,T_2^{10}$. We then calculate two relative errors $\forall i \in \{1, 2, \ldots, 10\}$,
\begin{align*}
      \epsilon^{T1-T2}_i =
\frac{\left|T_2^i - T_1(A)\right|}{T_1(A)}, \  \epsilon^{\lambda-T2}_i =
\frac{\left|T_2^i - \lambda(A)\right|}{\lambda(A)}.
\end{align*}
We plot the averages of $\{\epsilon^{T1-T2}_i\}$ and $\{\epsilon^{\lambda-T2}_i\}$ against the {\em actual number of nodes seen by the random walk}. Note that the x-axis accurately reflect how many times the algorithm actually queried the network, not just the number of samples used. Measuring the cost of uniform node sampling in this setting, for instance, would need to keep track of how many nodes are touched by a Metropolis-Hastings walk that implements the uniform distribution.

\subsection{Results}
Figure~\ref{fig:results} demonstrates the results.
For the two synthetic networks, the algorithm is able to get a 10\% approximation to the statistic $T_1(A)$ by exploring at most 10\% of the network. With more samples from the random walk, the mean relative errors settle to around 4-5\%.
However, once we measure the mean relative errors with respect to $\lambda(A)$, it becomes clearer that the estimator $T_2(A)$ does better when the graph is closer to the assumed (i.e. Chung-Lu) model. For the Chung-Lu graph, the mean error $\epsilon^{\lambda-T2}$ essentially is very similar to $\epsilon^{T1-T2}$, which is to be expected. For the preferential attachment graph too, it is clear that the estimate $T_2$ is able to achieve a better than $10\%$ relative error approximation of $\lambda(A)$.


Note that, if we were instead counting only the nodes whose degrees were actually used for estimation, the fraction of network used would be roughly $1-2\%$ in all the cases, the majority of the node cost actually goes in making the random walk mix. 

\begin{figure}[ht]
    \centering
        \includegraphics[width=0.45\textwidth]{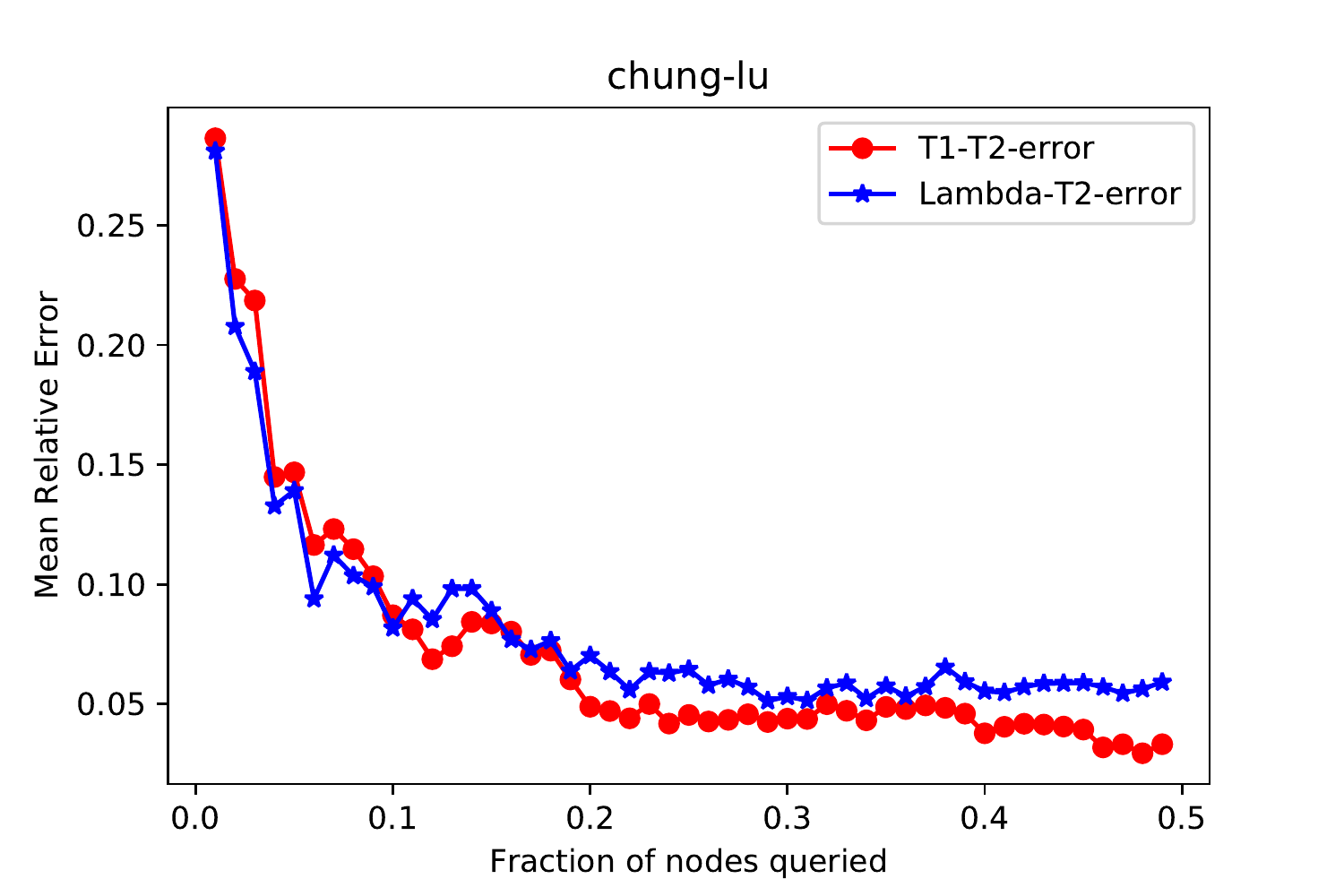}
    \includegraphics[width=0.45\textwidth]{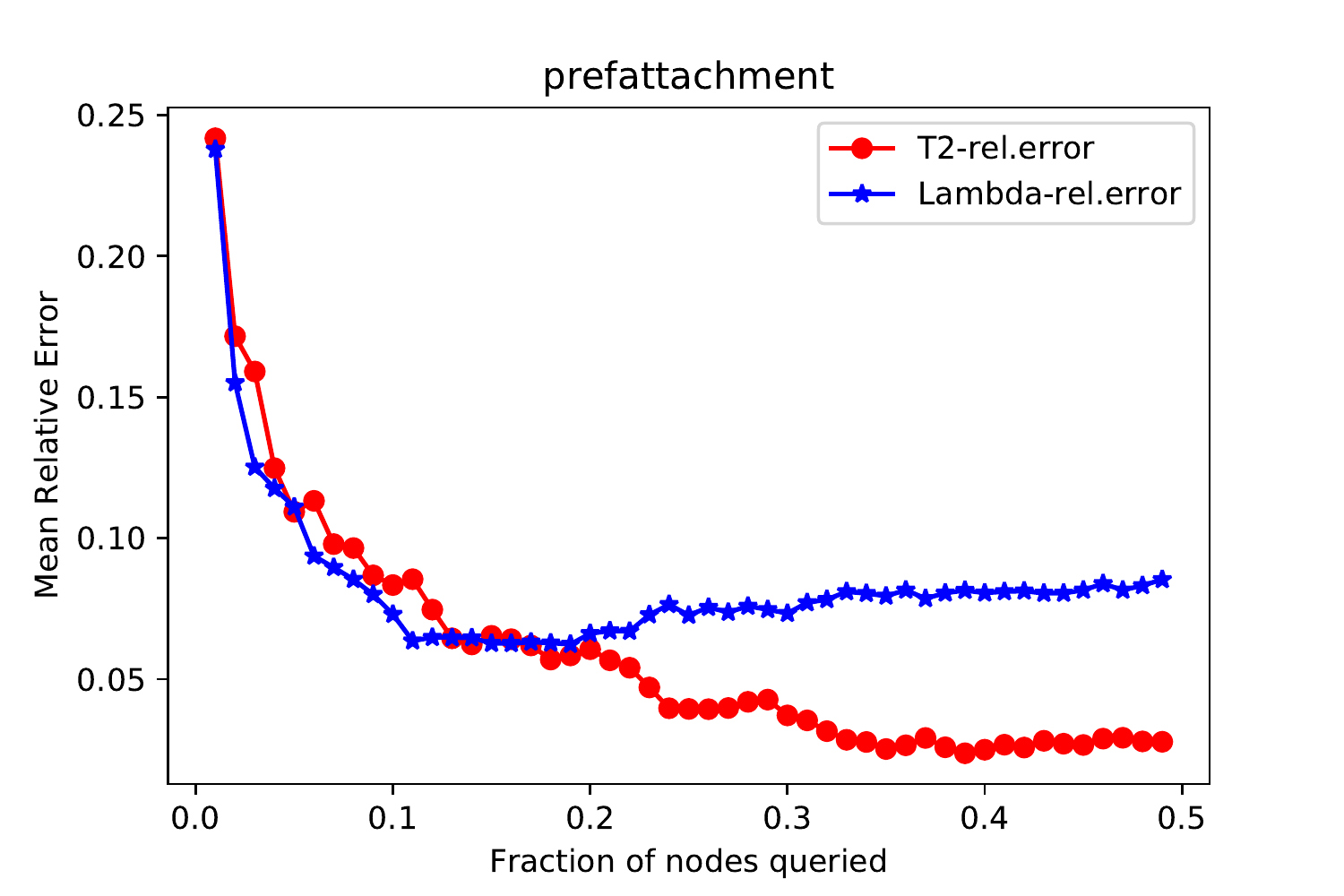}
    \caption{Results on two synthetic networks}
    \label{fig:results}
\end{figure}

\section{Discussion}
In this work, we investigated the problem of computing SIR epidemic thresholds of social contact networks from the perspective of statistical inference.
We considered the two challenges that arise in this context, due to high computational and data-collection complexity of the spectral radius.
For the Chung-Lu network generative model, the spectral radius can be characterized in terms of the degree moments.
We utilized this fact to develop two approximations of the spectral radius.
The first approximation is computationally efficient and statistically accurate, but requires data on observed degrees of all nodes.
The second approximation retains the computationally efficiency and statistically accuracy of the first approximation, while also reducing the number of queries or the sample size quite substantially.
The results seem very promising for networks arising from the Chung-Lu and preferential attachment generative models.

There are several interesting and important future directions.
The methods proposed in this paper have provable guarantees only under the Chung-Lu model, although it works very well under the preferential attachment model.
This seems to indicate that the degree based approximation might be applicable to a wider class of models.
On the other hand, this leaves open the question of developing a better ``model-free'' estimator, as well as asking similar questions about other network features. 

In this work we only considered the problem of accurate approximation of the epidemic threshold.
From a statistical as well as a real-world perspective, there are several related inference questions.
These include uncertainty quantification, confidence intervals, one-sample and two-sample testing, etc.

Social interaction  patterns vary dynamically over time, and such network dynamics can have significant impacts on the contagion process \cite{leitch2019toward}.
In this paper we only considered static social contact networks, and in future we hope to study epidemic thresholds for time-varying or dynamic networks.

\subsection{Disclaimer With Respect to Current Pandemic}

We do realize that in the face of the current pandemic, while it is important to pursue research relevant to it, it is also important to be responsible in following the proper scientific process. We would like to state that in this work, the question of epidemic threshold estimation has been formalized from a theoretical viewpoint in a much used, but simple, random graph model. We are not yet at a position to give any guarantees about the performance of our estimator in real social networks. We do hope, however, that the techniques developed here can be further refined to work to give reliable estimators in practical settings. 

\paragraph{Acknowledgements.} Anirban  acknowledges the kind support of the N. Rama Rao Chair Professorship at IIT Gandhinagar, the Google India AI/ML award (2020), Google Faculty Award (2015), and CISCO University Research Grant (2016).


\bibliographystyle{apalike}
\bibliography{ref,ref-sampling}
\end{document}